\documentclass[11pt,fullpage]{article}
\usepackage[margin=1in]{geometry}
\usepackage{graphicx} 
\usepackage{amsmath, amssymb, amsthm} 
\usepackage{mathtools, graphicx, enumitem} 
\usepackage{mathrsfs} 
\usepackage{hyperref} 
\usepackage{comment}
\usepackage{bbm}
\usepackage{xcolor}
\usepackage{natbib}
\usepackage{todonotes}
\usepackage{algpseudocode}
\usepackage{algorithm}

\definecolor{DarkGreen}{rgb}{0.1,0.5,0.1}
\definecolor{DarkRed}{rgb}{0.5,0.1,0.1}
\definecolor{DarkBlue}{rgb}{0.1,0.1,0.5}
\definecolor{Gray}{rgb}{0.2,0.2,0.2}

\hypersetup{
    unicode=false,          
    pdftoolbar=true,        
    pdfmenubar=true,        
    pdffitwindow=false,      
    pdfnewwindow=true,      
    colorlinks=true,       
    linkcolor=DarkBlue,          
    citecolor=DarkGreen,        
    filecolor=DarkRed,      
    urlcolor=DarkBlue,          
    pdftitle={},
    pdfauthor={},
}

\newcommand{\R}{\mathbb{R}}

\newcommand{\eps}{\varepsilon}

\newcommand{\OPT}{\mathrm{OPT}}
\newcommand{\paren}[1]{\left( #1 \right)}
\newcommand{\bracket}[1]{\left[ #1 \right]}

\DeclareMathOperator{\E}{\mathbb{E}}

\newcommand{\hU}{\hat{U}}

\newtheorem{theorem}{Theorem}[section]

\newtheorem{lemma}[theorem]{Lemma}
\newtheorem{proposition}[theorem]{Proposition}
\newtheorem{corollary}[theorem]{Corollary}
\theoremstyle{definition}
\newtheorem{definition}[theorem]{Definition}

\newcommand{\km}[1]{\todo[inline,color=teal!40]{\textbf{Kamesh:} #1}}

\makeatletter
\def\@fnsymbol#1{\ensuremath{\ifcase#1\or \dagger\or \ddagger\or
   \mathsection\or \mathparagraph\or \|\or **\or \dagger\dagger
   \or \ddagger\ddagger \else\@ctrerr\fi}}
\makeatother
\newcommand*\samethanks[1][\value{footnote}]{\footnotemark[#1]}

\usepackage[nameinlink,capitalize]{cleveref}

\title{Fair Multi-agent Persuasion with Submodular Constraints}
\author {
    Yannan Bai\thanks{Computer Science Department, Duke University, Durham NC 27708-0129. Email: {\tt \{yannan.bai,munagala,yiheng.shen,davidsonjiaduo.zhu\}@duke.edu}. This work is supported by NSF award IIS-2402823.} \and
    Kamesh Munagala\samethanks[1] \and
    Yiheng Shen\samethanks[1] \and
    Davidson Zhu\samethanks[1]
}
\date{}

\begin{document}

\maketitle

\begin{abstract}
    We study the problem of selection in the context of Bayesian persuasion. We are given multiple agents with hidden values (or quality scores), to whom resources must be allocated by a welfare-maximizing decision-maker. An intermediary with knowledge of the agents' values seeks to influence the outcome of the selection by designing informative signals and providing tie-breaking policies, so that when the receiver maximizes welfare over the resulting posteriors, the expected utilities of the agents (where utility is defined as allocation times value) achieve certain fairness properties. The fairness measure we will use is majorization, which simultaneously approximately maximizes all symmetric, monotone, concave functions of the utilities. We consider the general setting where the allocation to the agents needs to respect arbitrary submodular constraints, as given by the corresponding polymatroid. 
    
    We present a signaling policy that, under a mild bounded rationality assumption on the receiver, achieves a logarithmically approximate majorized policy in this setting. The approximation ratio is almost best possible, and that significantly outperforms generic results that only yield linear approximations. A key component of our result is a structural characterization showing that the vector of agent utilities for a given signaling policy defines the base polytope of a different polymatroid, a result that may be of independent interest. In addition, we show that an arbitrarily good additive approximation to this vector can be produced in (weakly) polynomial time via the multiplicative weights update method.
\end{abstract}

\section{Introduction}
\label{sec:intro}

The challenge of selecting fair outcomes arises in several decision-making settings, such as assembling project teams,  allocating institutional funding, and recommending items or articles. Consider for example a government agency allocating research funding across different research categories and institutions. The agency has a total budget, but may impose caps on funding allocated to any single institution or research area to encourage diversity.  Typically, the agency will be a welfare maximizer, allocating funding in a way that maximizes the average quality of the proposed work per dollar spent. However, the quality of proposed work is often hard to assess from the proposals, with several competing projects having comparable quality. The resulting uncertainty in assessing quality can create unintentional unfairness in allocating funds. 

One way to ameliorate this problem is to carefully design the proposal mechanism to reveal the right amount of additional information about the quality of the proposed work. Indeed, revealing too much can lead to an arbitrary winner-take-all allocation of funds, disadvantaging proposals from research areas or institutions that were only slightly inferior, while revealing too little causes many proposals, even of widely different quality, to be comparably ranked, again leading to low overall welfare, and hence unfairness. 

\paragraph{Signaling and Information Design.} This motivates the view of fair selection as an {\em information revelation} problem, an approach taken by \citet{banerjee_majorized_2024,AuKawai} in the context of selecting a single individual, for instance, in hiring or admission decisions. This information revelation problem is posed as a special case of Bayesian persuasion~\citep{Kamenica2011bayesian} as follows: The proposals in the above motivating example are viewed as ``agents''. Each agent has a quality that is drawn from an independent prior distribution. There is an intermediary that designs the information environment, in the above case, the proposing mechanism and what the agents should reveal in it. This intermediary (called the sender) is assumed to know the exact qualities. The decision-maker (called the receiver) is the reviewing panel in the above example and only knows the prior. The intermediary constructs signals independently for each agent based on the true qualities, and sends these to the receiver. Using these signals, the receiver constructs a posterior over the qualities, subsequently choosing the winning solution. In the above case, this is the budget allocation that maximizes posterior quality per dollar allocated, while respecting the budget constraints on research areas or institutions. Ties are broken by a randomized rule specified by the intermediary. 

The main research question then becomes: {\em How can an intermediary strategically reveal information about agents' qualities so that a welfare-maximizing decision-maker produces a fair outcome?} Such an approach to fairness via information revelation differs from prior algorithmic work on fair selection that designed novel, often randomized, selection rules~\citep{kleinberg2018selection,celis2020interventions,singh2021fairness,shen2023fairness,devic2024stability}.

Note that the decision-maker (receiver) always acts as a welfare maximizer over its information (the posterior), while the sender guides it towards socially desirable objectives such as fairness via partial information revelation (signaling). To gain intuition for why such signaling can be beneficial, in the funding allocation example, if the quality of each proposal were exactly known through an exhaustive questionnaire, then the decision-maker could identify the absolute best candidates to fund. However, such detailed information forces the decision-maker to make an arbitrary, winner-take-all choice among comparable proposals. In contrast, if we carefully design either the questionnaire or  review process to be coarser (for instance,  by letting outside expert reviewers rate proposals as ``competitive'' or ``not competitive''), this could render many high-quality proposals indistinguishable to the decision-maker.  This can let them perform randomized tie-breaking that gives each high-quality proposal a fair chance, hence ensuring fairness while not compromising overall welfare significantly. The question then becomes -- how should this signal (the questionnaire or review process in this case) be designed with the fairness and welfare goals in mind?

\paragraph{Submodular Constraints and Polymatroids.} In this work, we generalize the information design framework of~\citet{banerjee_majorized_2024} beyond selecting a single agent to the broader domain of selection when the allocations define a polymatroid~\citep{schrijver_2003,fujishige2005submodular}. Here, the decision-maker's goal is to choose a feasible set (possibly involving fractional or randomized allocations) that maximizes total welfare. Polymatroids capture submodularity in the allocation constraints, and hence the structure of many selection problems involving diminishing returns or diversification. In the funding allocation example, budget caps on subsets of research areas or institutions, if they are submodular, will define a polymatroid.

A special case of polymatroids is {\em matroid constraints}, that is, allocations that are a randomization over independent sets of a matroid. This includes choosing one candidate in hiring or admission decisions considered in~\citet{banerjee_majorized_2024,AuKawai}, and its generalization to selecting $k$ candidates. Similarly, it includes partition matroids, where agents are partitioned into disjoint groups with a quota for each group. For example, in the task of hiring a specialized team of $n$ engineers, there could be constraints that at most $k_1$ test engineers and at most $k_2$ security experts are selected. As before, the agents' quality is partially known to the decision-maker, and the agents could signal their quality via an information intermediary, with the goal of ensuring fair selection.

As a different example, consider a recommendation platform that presents a ranked list of articles or products to a user. Each item has a known click probability—capturing its visibility when shown in a particular position—but an unknown \emph{quality} (e.g., user satisfaction conditional on click). We assume that a user, when shown an ordering of the items, scans linearly, stopping at the first item they click on, and selecting (reading/purchasing) that item. The vector of expected selection probabilities of the items for different orderings defines a polymatroid~\citep{DBLP:journals/mor/DeanGV08}. Suppose the platform is welfare maximizing and orders items to maximize expected quality of the selected item. The platform may not directly observe true quality, which can depend on external information such as expert reviews, advertiser relevance estimates, or cross-platform user data. An information intermediary, such as an ad exchange or third-party curator, may possess a more accurate signal of each item’s quality from its own predictive models and data. By selectively revealing or coarsening this information to the platform—for instance, through relevance scores or category-level ratings, the intermediary can influence the platform’s ranking so that the induced click probabilities and expected utilities are more balanced across items.

\paragraph{Remark on Information Intermediaries.} In both motivating examples, the  intermediary plays the role of a mechanism or entity that shapes the information available to the decision-maker. In settings such as hiring or grant evaluation, this corresponds to the design of questionnaires, review rubrics, or external expert ratings that determine how much detail about candidates’ or proposals’ true quality is revealed to
the decision makers. Intermediaries also naturally arise as components such as ad exchanges or content-curation models that possess richer side  information via predictive models and transmit these to the recommendation platform as described above. In all cases, the intermediary does not directly allocate resources, but influences allocation outcomes by deciding how coarsely or finely to reveal private information. 


\paragraph{Fairness as Majorization.} 
To formalize the goal of fairness, we adopt the notion of {\em approximate majorization} from~\citet{banerjee_majorized_2024}. This notion arises from economics and operations research~\citep{hardy1934inequalities,goel2005approximate,kumar2006fairness,chakrabarty2019approximation}. Given a signaling policy, we assume each agent's utility equals its quality score times its fractional allocation.\footnote{Our results easily generalize to the setting where agent $i$'s utility is $v x_i$, where $v$ is a fixed multiplier and $x_i$ is the fractional allocation to the agent.} Each agent therefore receives an expected utility, where the expectation is over its quality and the outcome of the signaling policy. Our signaling policies will ensure the resulting vector of agents' expected utilities is approximately majorized; see \cref{def:majorize} for a formal definition. Informally, this means all symmetric, monotone, and concave fairness functions (including max-min fairness, total welfare, and Nash welfare) are simultaneously optimal to that approximation factor. In particular, such a notion is approximate on both traditional fairness notions (like max-min fairness) and the total welfare, in some sense being approximately best possible. Our goal in this paper is to design a signaling policy that achieves as small an approximation factor as possible, in a computationally efficient fashion.

As mentioned before, the key challenge in designing a fair signaling policy lies in the trade-off between too little and too much information. Decisions made with no additional signaling are unlikely to be fair. On the other hand, complete information revelation can also produce unfair outcomes, since comparable individuals can be treated very differently. We illustrate this trade-off via an example in \cref{app:example}. Our work initiates the study of fair information design for  selection problems with complex allocation constraints, specifically with polymatroid constraints.

\subsection{Main Results}
\label{sec:contri}
Our main technical result is the following  theorem about existence of approximately majorized solutions, which generalizes an analogous result in~\citet{banerjee_majorized_2024} from selecting a single agent to handle arbitrary submodular constraints:

\begin{theorem}[Informal]
\label{thm:main-informal-exist}
Assuming there are $n$ agents and each agent's quality lies in $[1,V]$, there is a  $O\left(\log \frac{V}{\epsilon} \right)$ approximate majorized policy when:
\begin{itemize}
    \item the agents have independent distributions over quality (or value) and independent signaling policies for this quality (see \cref{sec:prelim} for a detailed model), 
    \item the set of feasible allocations over the agents forms a polymatroid,  
    \item the receiver is a $(1 + \epsilon)$-approximate (in each coordinate) welfare maximizer over the polymatroid constraint given the posterior distributions over agent values, and
    \item the utility of an agent is the expectation over the signaling policy of its value times its expected fractional allocation given that value.
\end{itemize} 
\end{theorem}

The example in \cref{app:example} shows that naive signaling policies cannot achieve the above bound even with simple polymatroidal constraints. We further note that there is a lower bound of $\Omega(\log \log V)$ on approximate majorization even for selecting a single agent~\citep{banerjee_majorized_2024}. This rules out an $O(1)$ approximate majorized signaling policy. 

We next complement \cref{thm:main-informal-exist} by showing that an arbitrarily good approximation to the above solution can be computed in (weakly) polynomial time via an application of the multiplicative weights method.

\begin{theorem}[Informal]
\label{thm:main-informal-compute}
For any $\delta > 0$, the utility vector from \cref{thm:main-informal-exist} can be approximated to an additive $O(\delta)$ in time polynomial in $1/\delta$, $n$, and $V$. 
\end{theorem}

\paragraph{Remarks.}  The notion of a $(1+\varepsilon)$-approximate receiver in \cref{thm:main-informal-exist}  is key to our approach for handling general signaling policies.\footnote{We note that polynomial-time-computability results of~\citet{dughmi2016algorithmic} for Bayesian persuasion also assume an approximately optimal receiver.} It models a decision-maker who acts on a simplified information space: they first categorize agents' posterior means into a finite set of buckets and then optimize welfare based on a canonical value for each bucket. This is a natural model of bounded rationality, reflecting how decision-makers often simplify complex continuous inputs.  (See \cref{sec:approx} for details.)

We also note that the generic result of \citet{goel_simultaneous_2006}  gives an $O\left(\min\left\{n,\log \frac{U_{\max}}{n \cdot U_{\min}}\right\}\right)$-majorized solution, where $U_{\min}$ is the max-min fair utility, and $U_{\max}$ is the social welfare. Since these utilities depend on the allocation polytope and the distributions of quality scores, their ratio can be exponentially large, leading to an approximation ratio of $n$ in the worst case. In contrast, our approximation factor from \cref{thm:main-informal-exist} only depends on the scale of the quality scores, and can be much smaller.

We finally note that \cref{thm:main-informal-exist} requires the polymatroidal structure of the constraint set. In \cref{app:non-polymatroid-impossibility}, we show an example with non-polymatroidal constraints where there is no sub-linear (in number of agents $n$) approximation to majorization when the underlying allocation set is not a polymatroid. This holds even when $V=1$ and the agent values are deterministic so that no signaling is required.

\subsection{Technical Contributions} Our proof proceeds in two main stages. In \cref{sec:full-rev}, we first analyze the simpler but crucial case of ``full revelation'' policies, which means each agent truthfully reveals its quality. We do so to establish a novel structural property of the utility space, which we use as a building block to prove our main theorem for general signaling policies in \cref{sec:approx}. 

\paragraph{Existence Result.} To show the existence result in \cref{thm:main-informal-exist}, in \cref{sec:full-rev}, we start with the simple setting where the signaling policy is ``full revelation''.  In this setting, given the revealed values of the agents, the receiver has a choice between welfare maximizing allocations and chooses a solution for each vector of revealed qualities so that in expectation over these revelations, the vector of agent utilities is as fair as possible. Our main result is \cref{thm:maj-policy} that shows the existence of an exactly majorized policy in this setting. The key to showing this result is \cref{lem:mat-poly}, which provides a polyhedral characterization of persuasion with a welfare maximizing receiver:

\begin{quote}
    When the underlying feasibility constraint over agent allocations defines a polymatroid and the receiver is a welfare maximizer, the set of expected utility vectors of the agents is also the base polytope of a (different) polymatroid.
\end{quote} 

Given this statement, we can leverage the existence of exactly majorized solutions for polymatroids~\citep{tamir_least_1995,veinott_least_1971,megiddo_optimal_1974}. To see why this statement is non-trivial, we note that in each scenario of revealed values, set of welfare maximizing allocations define a face of the polymatroidal extension, which is itself the base polytope of a polymatroid~\citep{schrijver_2003,gijswijt2010caratheodoryrankpolymatroidbases}.   However, in our setting, the utility vector of the agents is not the allocation vector, but the allocation scaled by the quality of each agent, and in general, even if the allocation vectors are drawn from a polymatroid, but are scaled by fixed quantities that depend on the index of the coordinate, the resulting vectors {\em do not} define a polymatroid. This makes the statement of \cref{lem:mat-poly} novel and non-trivial; we provide additional discussion for its subtlety in \cref{app:sec3}. Our proof of \cref{lem:mat-poly} crucially uses the welfare maximizing behavior of the receiver and the polymatroid structure of the underlying feasibility constraint. We carefully analyze the greedy allocation rule of welfare optimization to show polymatroidal structure of the overall signaling problem. In \cref{app:sec3}, we complement this result by showing an example where the allocation set is not a polymatroid but has a $1$-majorized point, but the utility vectors cannot be approximately majorized. This showcases the crucial role of submodularity in our results.

Once we establish exact majorization for full revelation policies, in \cref{sec:approx},  we combine this with the idea of ``single-mean policies'' from~\citet{banerjee_majorized_2024} and the existence of exactly majorized points for polymatroids~\citep{tamir_least_1995} to show the existence of a logarithmically majorized policies for general signaling policies. Note that unlike the full revelation setting, it is no longer possible to show polymatroidal structure for the space of utilities in general signaling policies. A key ingredient here is \cref{lem:maximal-mapping-dominance} that shows the existence of a single-mean policy of a certain type that does have polymatroid structure, again using the monotonicity and submodularity properties of the underlying polymatroid. This step crucially requires the agents have independent distributions over their quality and independent signaling policies.

\paragraph{Computation.} The ideas for showing \cref{thm:main-informal-compute} are more standard. For this, in \cref{sec:poly-time}, we write majorization problem as a set of linear programs over the exponentially many scenarios of revelations, with a polymatroid optimization problem (capturing receiver behavior) for each scenario. We use a majorization LP from~\citet{goel_simultaneous_2006}. We now use the multiplicative weights method~\citep{arora_multiplicative_nodate} to solve this program approximately; the dual oracle becomes the expectation over scenarios, of a weighted welfare maximization problem over the base polytope of the polymatroid capturing the receiver optimization problem for that scenario. The latter can again be solved via a greedy algorithm, and the expectation can be approximated by sampling polynomially many scenarios. Our overall approach follows~\citet{cai_optimal_2012,bhalgat_optimal_2013}, who apply similar frameworks for optimal multidimensional auctions. 

\paragraph{Conceptual Contribution.} Conceptually, \cref{lem:mat-poly,lem:maximal-mapping-dominance} allow for a direct analysis of the utility space. This is in contrast to the approach in~\citet{banerjee_majorized_2024}, which focused on the special case of single-agent selection (a matroid special case). Their work  reduced a relaxed version of the single-selection problem to a majorized network flow instance. In contrast, our proof is more direct, and provides a new structural understanding via a geometric and polymatroidal characterization of the utility space. This not only allows us to leverage existing results on finding majorized points in polymatroids~\citep{tamir_least_1995} but also provides hope of achieving fairness and welfare guarantees for information design in other, more complex, settings. 

At a higher level, the difficulty with persuasion is that the sender needs to treat the receiver's optimization routine as a black box, which makes the optimal persuasion problem non-convex in general, even when the receiver is solving a convex optimization problem. This aspect has precluded the development of general-purpose techniques based on convex relaxations to derive structural insights into these problems. As an example, optimal auction design and pricing under persuasion requires the development of specialized techniques to handle non-convexity, and some generalizations admit strong lower bounds for this reason~\citep{AlijaniBMW22,Banerjee2024fair}. Our main contribution is to show that a large class of submodular selection problems admits to a convex structure even in the presence of persuasion, a result that is {\em a priori} not obvious.

\subsection{Other Related Work}
\label{sec:work}
\paragraph{Bayesian Persuasion.}
\emph{Information design} is a framework for understanding how a sender can influence a receiver's actions by strategically revealing information~\citep{bergemann2019information, dughmi2017algorithmic}.  Within this broad area, our work falls in the setting of Bayesian persuasion~\citep{Kamenica2011bayesian}, where the receiver performs Bayesian updates based on the sender's signals. This problem has been widely studied in various contexts in computer science and economics~\citep{bergemann2015limits,Banerjee2024fair,DBLP:conf/aaai/XuRDT15,DBLP:conf/sigecom/BabichenkoTXZ21,DBLP:journals/mktsci/ChakrabortyH14,tang2024intrinsic}.  We note that computationally efficient algorithms exist for arbitrary objectives in persuasion, notably the FPTAS of \citet{dughmi2016algorithmic}. However, majorization requires the simultaneous near-optimality of an entire family of fairness functions, and here, even showing existence is non-trivial. Making progress therefore requires the development of new structural insights into the problems that deviate significantly from prior literature. As mentioned before, our work derives novel convexity characterizations for a large class of persuasion problems, allowing us to argue strong fairness properties.

Our model builds on work by~\citet{AuKawai}, who consider selfish agents who independently construct their own signaling policies to persuade a receiver to allocate to them. In contrast with their work that focuses on allocating to one agent, we consider general allocation polymatroids, and focus on fairness in a centralized setting with a common sender, akin to~\citet{banerjee_majorized_2024}.

\paragraph{Majorization in Optimization.}
Majorization was introduced in the seminal works of~\citet{karamata1932inegalite,hardy1934inequalities}. It provides a strong framework for fairness that is equivalent to maximizing all symmetric and concave welfare functions. The work of~\citet{goel2005approximate,goel_simultaneous_2006} defined an approximate version suitable for resource allocation, hence applying it to approximation algorithms. More classical work has found connections between majorization and specific combinatorial structures, most notably the exact majorization of flows in single-source multi-sink networks~\citep{veinott_least_1971,megiddo_optimal_1974}, and of polymatroids in general~\citep{tamir_least_1995}. Our work contributes to this literature by establishing a new structural connection between majorization in Bayesian persuasion and the geometry of polymatroids. The resulting approximation ratios are a significant improvement over generic approximation bounds that follow from~\citet{goel_simultaneous_2006}, that can depend linearly on problem parameters.

\section{Preliminaries}\label{sec:prelim}

\subsection{Signaling Policies} 
There is a set $E$ of $n$ agents. We call the decision-maker the receiver. The value $v_i$ of each agent $i$ is drawn independently from the distribution $D_i$. The decision-maker knows the distributions $\{D_i\}_{i=1}^n$, but does not know the realized values $\{v_i\}_{i=1}^n$. We assume $v_i \in [1,V]$ for all $i$.

After the values $\{v_i\}_{i=1}^n$ are realized, an intermediary (or sender) uses these values to send signals $\{\sigma_i\}_{i=1}^n$ to the receiver via a \textit{signaling policy}. A signaling policy $\omega$ comprises the \textit{mapping rule} and the \textit{selection rule}. The signal for each agent is independently constructed from the other agents.

\paragraph{Mapping Rule.}  A mapping rule is a collection of signals for each agent $\{\Gamma_i\}_{i=1}^n$ together with a function that maps the value $v_i$ of an agent $i$ to a distribution $g_{iv_i}$ over signals in $\Gamma_i$. When the sender sees the realized values $\{v_i\}_{i=1}^n$, they compute the corresponding signal distributions $\{g_{iv_i}\}_{i=1}^n$ by the mapping rule. They then generate the realized signals $\{\sigma_i\}_{i=1}^n$ by drawing $\sigma_i\sim g_{iv_i}$ independently for each agent, and the receiver sees $\{\sigma_i\}_{i=1}^n$.  

After receiving the set of signals, the receiver computes the posterior distributions $\{D_i(\sigma_i)\}_{i=1}^n$ over agent values using Bayes' rule. Let $\mu_i = \E[D_i(\sigma_i)]$ denote the posterior mean of agent $i$. A set of posterior means is said to be {\em Bayes plausible} if it corresponds to a valid signaling policy. Under Bayes plausibility, the expectation of the posterior mean over the signals is equal to the prior mean.

\paragraph{Allocation Constraints.} There is a polymatroid constraint $(E,f)$ on the set of possible allocations to the agents.

\begin{definition}[Submodularity and Polymatroids]
    Let $E$ be a finite set and $f$ a non-negative, monotone, submodular, function from the power set $2^E$ to $\R_+$, which satisfies 
 $f(\emptyset)=0$; 
 $$f(A)\leq f(B) \mbox{ for } A\subseteq B\subseteq E, \mbox{ and }$$ 
     $$ f(A)+f(B)\geq f(A\cup B)+f(A\cap B) \mbox{ for } A,B\subseteq E.$$

    Then, the pair $(E,f)$ is called a polymatroid, where $E$ is called the ground set and $f$ the rank function of the polymatroid. A polymatroid defines a polytope $\mathcal{P}(f)\subset \R_+^E$ by \[\mathcal{P}(f)=\{\mathbf{x}: \mathbf{x}(A)\leq f(A)\text{ for all }A\subseteq E\}.\]
    This polytope is called the independence polytope of the polymatroid. When there is no ambiguity, we also refer to the independence polytope as the polymatroid. 
    
    The {\em base polytope} of a polymatroid (or the corresponding submodular function) is the following:
$$ \mathcal{B}(f) = \{ \mathbf{x} \in \mathcal{P}(f): \mathbf{x}(E) = f(E)\}.  $$
\end{definition}

\paragraph{Selection Rule.} We assume that the receiver is a utilitarian welfare maximizer, so that it maximizes the sum of the posterior utilities of the agents subject to the polymatroid constraint $(E,f)$.  In other words, the receiver constructs a welfare-maximizing allocation $\mathbf{x} \in \mathcal{P}(f)$ that maximizes $\sum_{i = 1}^n \mu_i x_i$. The set of welfare-optimal allocations $\mathbf{x}$ defines a face of the polymatroid $(E,f)$, and the receiver chooses an allocation from this face to satisfy the sender's auxiliary objective (described later). This choice is termed the {\em selection rule}.  We assume agent $i$ obtains utility $\mu_i x_i$. 

A signaling policy $\Omega$ is a distribution over independent signaling policies $\omega$. Before the process starts, the sender draws $\omega\sim \Omega$ and implements $\omega$.  The signaling policy is known to the receiver. Since  $v_i \sim D_i$, this yields a expected utility $U_i(\Omega)$ for the agent, where the expectation is over $D_i$, the distributions of other agents' values, and the distribution over signaling policies in $\Omega$.  

\subsection{Fairness and Majorization} 
The goal of the sender is to design a signaling policy $\Omega$ that is fair. We capture this as designing $\Omega$ such that the vector $\{U_i(\Omega)\}_{i=1}^n$ is $\alpha$-majorized over the set of all signaling policies, for the smallest possible value $\alpha$. The selection rule of $\mathbf{x}$ among the receiver's welfare-maximizing allocations will be influenced by this fairness goal.

\begin{definition}[$\alpha$-Majorization,~\citet{goel_simultaneous_2006,banerjee_majorized_2024}]
\label{def:majorize}
For $\alpha \ge 1$, a signaling policy $\Omega$ is called \emph{$\alpha$-majorized}  if for any $k \in \{1,2,\ldots,n\}$ and any signaling policy $\Omega'$, the sum of the $k$ smallest utilities in $\{U_i(\Omega)\}_{i=1}^n$ is at least $1/\alpha$ times the sum of the $k$ smallest utilities in $\{U_i(\Omega')\}_{i=1}^n$.
\end{definition}

The following result shows that approximate majorization is equivalent to simultaneously approximating all symmetric and concave welfare functions.

\begin{proposition} [\citet{goel_simultaneous_2006}]
The signaling policy $\Omega$ is \emph{$\alpha$-majorized} if and only if for every symmetric, non-decreasing, and concave function $h \colon \mathbb{R}_{\geq 0}^n \to \mathbb{R}_{\geq 0}$  and any other signaling policy $\Omega'$, 
\[
h\left( U(\Omega) \right) \ge \frac{1}{\alpha} \cdot h\left( U(\Omega') \right).
\]
\end{proposition}

As mentioned before, we illustrate why signaling can improve fairness via an example in \cref{app:example}.

\subsection{Properties of Polymatroids}


\paragraph{Greedy Algorithm.} Given a polymatroid $\mathcal{P}(f)$, the classic greedy algorithm works as follows:
\begin{enumerate}
    \item Order the indices of $[n]$ according to a permutation (ordering) $\pi$.
    \item  For $k=1, \dots, n$, set
    $ x_{\pi(k)} = f(\{\pi(1), \dots, \pi(k)\}) - f(\{\pi(1), \dots, \pi(k-1) \}) $.
\end{enumerate}

We have the following well-known lemmas:

\begin{lemma}[\citet{schrijver_2003}, Chapter 44] 
\label{lem:greedy1}
Given a vector $\vec{v} > \vec{0}$, the function $\vec{v} \cdot \vec{x}$ is maximized over $\mathcal{P}(f)$ by the greedy algorithm that uses an ordering $\pi$ such that $v_{\pi(1)} \ge v_{\pi(2)} \ldots \ge v_{\pi(n)}$. Further, any vertex\footnote{This is shown for the extended polymatroid in~\citet{schrijver_2003}. Note that for a strictly positive $\vec{v}$, the optimum face will belong to the extended polymatroid.} of the optimal face of $\mathcal{P}(f)$  corresponds to some permutation $\pi$ satisfying $v_{\pi(1)} \ge v_{\pi(2)} \ldots \ge v_{\pi(n)}$.
\end{lemma}

\begin{lemma}[\citet{fujishige2005submodular}, Chapter 3] 
\label{lem:greedy2}
The set of vertices of $\mathcal{B}(f)$ coincides with the set of vectors $\vec{x}$ obtained by running the greedy algorithm for all possible orderings $\pi$.
\end{lemma}

In the classic lemma below, the result for the polytopes $\mathcal{P}$ is from~\citet{schrijver_2003}, while the result for base polytopes uses the above characterization of its vertices --- it is easy to write each vertex of $\mathcal{B}(f_1 + f_2) $ as the sum of the corresponding vertices of $\mathcal{B}(f_1) + \mathcal{B}(f_2)$.

\begin{lemma}[\citet{schrijver_2003}]
\label{cor:mon-sub}
The following statements hold for independence and base polyhedra of non-negative, monotone, submodular functions:
\begin{itemize}
\item 
$\mathcal{P}(f_1) + \mathcal{P}(f_2) = \mathcal{P}(f_1 + f_2) \qquad \mbox{and}  \qquad \mathcal{B}(f_1) + \mathcal{B}(f_2) = \mathcal{B}(f_1 + f_2);$ 
\item For any $\alpha > 0$, we have
$\qquad \mathcal{P}(\alpha f) = \alpha \mathcal{P}(f)  \qquad \mbox{and}  \qquad \mathcal{B}(\alpha f) = \alpha \mathcal{B}(f).$
\end{itemize}
\end{lemma}

\paragraph{Majorization.} The following lemma captures the relation between polymatroids and majorization. 

\begin{lemma}[\citet{tamir_least_1995}]\label{thm:poly-major}
    Any polymatroid $(E,\rho)$ has a $1$-majorized element that lies in $\mathcal{B}(\rho)$. 
\end{lemma}

\section{Full Revelation Signaling Policies}
\label{sec:full-rev}

We first show the existence of $1$-majorized solution for the more restricted \textit{Full Revelation Policies}, and show an approximation algorithm to compute it. This will form the basis of the proof of approximate majorization (and the associated computational result) for general policies in \cref{sec:approx}. 

The mapping rule of a full revelation policy is directly sending the realized value to the receiver, and the signaling policy involves designing the selection rule for the receiver. Assume that there are only finitely many possible values for all agents, denoted by $v_1 > v_2 >\cdots >v_k > 0$. Let $v_{a_i}$ be the value sent by agent $i$. In this setting, note that if $v_{a_i}$ is the revealed value of agent $i$, then the posterior mean is simply $\mu_i = v_{a_i}$, and the receiver selects an $\mathbf{x} \in \mathcal{P}(f)$ that maximizes $\sum_i v_{a_i} x_i$, breaking ties in favor of the majorization objective. This yields expected utility vector $\{U_i(\Omega)\}_{i=1}^n$ for the agents.

\subsection{Existence of a 1-Majorized Solution}\label{sec:exist}
We will show the following theorem.
\begin{theorem}[Existence of $1$-majorized policy]\label{thm:maj-policy}
    Assume the allocation constraints define a polymatroid $(E,f)$. For the class of full-revelation signaling policies, the set of expected utility vectors $\{U_i(\Omega)\}_{i=1}^n$ for signaling policies $\Omega$ has a $1$-majorized point. 
\end{theorem}

Fix a set of realized values of the agents. We will show that the vector of utilities of the agents forms the base polytope of a different polymatroid. 


\begin{lemma}
\label{lem:mat-poly}
Let $\mathcal{P} = \mathcal{P}(f)$ be a polymatroid on a ground set $E$, defined by a submodular rank function $f$. Let $\Vec{v}$ be a strictly positive vector of agent values. Let $\mathcal{X}^*$ be the face of optimal allocations in $\mathcal{P}$ that maximize the welfare function $v \cdot x$. Let the corresponding set of utility vectors be $\mathcal{U} = \{ (v_i x_i)_{i \in E} \mid x \in \mathcal{X}^* \}$. Then, the set $\mathcal{U}$ is the base polytope of a submodular function.
\end{lemma}
\begin{proof}
We define the saturation function $g: 2^E \to \mathbb{R}_+$ associated with $\mathcal{U}$ as:
\[ g(S) = \max_{u \in \mathcal{U}} \sum_{i \in S} u_i = \max_{x \in \mathcal{X}^*} \sum_{i \in S} v_i x_i, \]
where $v_i$ is the $i^{th}$ coordinate of $\Vec{v}$ and $u_i = v_i x_i$ is the $i^{th}$ coordinate of $u$. This function is trivially monotone. We will show that $g$ is submodular. We do so by deriving a closed-form expression for $g(S)$. Subsequently, we will show that $\mathcal{U} = \mathcal{B}(g)$, completing the proof.

By \cref{lem:greedy1}, the vertices of the optimal face $\mathcal{X}^*$ are generated by the polymatroid greedy algorithm for the objective vector $\Vec{v}$, with different outcomes arising from different tie-breaking orders for agents with the same value $v_i$. The value $g(S)$ is therefore achieved at the vertex obtained by running the greedy algorithm with a tie-breaking rule that prioritizes maximizing the utility from the set $S$.

Let the agents $E$ be partitioned into blocks $E_1, E_2, \dots, E_k$ where all agents in a block $E_j$ have the same value $v_j$, and $v_1 > v_2 > \dots > v_k > 0$. The greedy algorithm proceeds through these blocks sequentially. To find the value of $g(S)$, we define a specific permutation $\pi_S$ as follows: For each block $E_j$, agents in $S_j = S \cap E_j$ are processed first. Agents in $E_j \setminus S_j$ are processed next. Within these subsets, any fixed arbitrary order is used. The allocation vector for this permutation is $x(\pi_S)$, and $g(S) = \sum_{i \in S} v_i x_i(\pi_S)$. The total utility from the agents in $S$ is therefore:
\[ g(S) = \sum_{j=1}^k \sum_{i \in S_j} v_j x_i(\pi_S) = \sum_{j=1}^k v_j \left( \sum_{i \in S_j} x_i(\pi_S) \right). \]
We now analyze the inner sum for a single block $j$. Let $P_{<j} = E_1 \cup \dots \cup E_{j-1}$ be the set of all agents in higher-value blocks. Let the processing order for agents in $S_j$ be $s_1, s_2, \dots, s_m$. Using the greedy allocation for the priority rule discussed above, we have:
\begin{align*}
\sum_{i \in S_j} x_i(\pi_S) & = \sum_{j=1}^m x(s_j) = \sum_{j=1}^m  \left(f(P_{<j} \cup \{s_1, \ldots,s_j\}) - f(P_{<j} \cup \{s_1, \ldots, s_{j-1}\}) \right) \\
& = f(P_{<j} \cup S_j) - f(P_{<j}).
\end{align*}

Substituting this back into the expression for $g(S)$, we arrive at the closed-form formula:
\[ g(S) = \sum_{j=1}^k v_j \left[ f(P_{<j} \cup (S \cap E_j)) - f(P_{<j}) \right]. \]

 Let $g_j(S) = v_j [ f(P_{<j} \cup (S \cap E_j)) - f(P_{<j}) ]$. To show this function is submodular, we only need to note that $h_j(S) = f(P_{<j} \cup (S \cap E_j))$ is submodular.  Since each $g_j(S)$ is submodular, their sum $g(S) = \sum_j g_j(S)$ is also submodular.

Let $\mathcal{B}(g)$ now denote the base polytope of the polymatroid with rank function $g$. We will now show that $\mathcal{U} = \mathcal{B}(g)$, completing the proof. 

First, by definition, $\mathcal{X}^*$ is the set of allocations $x \in \mathcal{P}$ that maximize $v \cdot x$.  Let this maximum welfare be $W_{\max}$. Thus, for any $x \in \mathcal{X}^*$, we have $v \cdot x = W_{\max}$. Now, consider any $u \in \mathcal{U}$. By definition, $u_i = v_i x_i$ for some $x \in \mathcal{X}^*$, so that  $\mathcal{U}$ lies on the hyperplane $\{z \in \mathbb{R}^n \mid \sum z_i = W_{\max}\}$.  Further, for any $u \in \mathcal{U}$ and any $S \subseteq E$, we have $\sum_{i \in S} u_i \le g(S)$ by definition of $g(S)$. These two observations imply $\mathcal{U} \subseteq \mathcal{B}(g)$.  Further,  $\mathcal{U}$ is convex, since $\mathcal{X}^*$ is a face of $\mathcal{P}$ and is hence  convex. 

We will finally show that $\mathcal{B}(g) \subseteq \mathcal{U}$ by showing that any vertex of $\mathcal{B}(g)$ corresponds to a feasible realization of utilities; since $\mathcal{U}$ is convex, this will imply any interior point of $\mathcal{B}(g)$ also lies in $\mathcal{U}$. By \cref{lem:greedy2}, any vertex of $\mathcal{B}(g)$ can be obtained by ordering the elements of $E$; relabel them as $1,2,\ldots,n$ in this ordering, and setting $u_i = g([i]) - g([i-1])$. Let $S = [i-1]$, and let $i \in E_j$ as defined above. When we add $i$ to $S$, it can be checked that $u_i = v_j \cdot \left( f(P_{<j} \cup \{i\} \cup (S \cap E_j)) -  f(P_{<j} \cup (S \cap E_j)) \right)$. This corresponds to the receiver assigning $x_i = f(P_{<j} \cup \{i\} \cup (S \cap E_j)) -  f(P_{<j} \cup (S \cap E_j))$, that is, placing $i$ next in the tie-break ordering for $E_j$ after the elements of $P_{<j}$ and $S \cap E_j$, and allocating greedily. Therefore, any vertex of $\mathcal{B}(g)$ corresponds to a feasible realization of utilities by some tie-breaking rule of the receiver. This implies $\mathcal{B}(g) \subseteq \mathcal{U}$. Therefore, $\mathcal{U}$ is the base polytope $\mathcal{B}(g)$ for the submodular function $g$.
\end{proof}

\begin{proof}[Proof of \cref{thm:maj-policy}]
Note that $\{U_i(\Omega)\}_{i=1}^n$ is the Minkowski sum of the vectors $\{\Pr[\sigma] \cdot U_i(\sigma)\}$, where $\sigma$ are scenarios of realized values.  By \cref{cor:mon-sub}, scaling a base polytope by a constant is also a base polytope, and so is taking Minkowski sums. Combining with \cref{lem:mat-poly}, this means the set $\{U_i(\Omega)\}_{i=1}^n$ also defines the base polytope of a polymatroid, and by \cref{thm:poly-major}, this has a $1$-majorized point. 
\end{proof}

Indeed, by combining \cref{cor:mon-sub} with \cref{lem:mat-poly}, the set of vectors  of expected utilities $\vec{U}$ of the agents for feasible signaling policies coincides with the base polytope of the following polymatroid $\mathcal{R}$. Here, $g(S;\vec{v})$ is the function $g(S)$ from \cref{lem:mat-poly} when the realized value vector is $\vec{v}$. The expectation below is over the realized value vector.
\begin{equation}
\label{eq:poly1}
\mathcal{R} = \left\{ \vec{y} \ge 0\  \big| \ \sum_{i \in S} y_i \le \E_{\vec{v}}\left[g(S;\vec{v})\right] \ \forall S \subseteq [n]\ \right\}.
\end{equation}
\cref{lem:mat-poly} implies $\E[g(S;v)]$ is submodular, so that the above set of constraints define a polymatroid, and has a $1$-majorized point.

\paragraph{Remarks.} We note that the proof of \cref{lem:mat-poly} is quite delicate. In \cref{app:sec3}, we present two examples to support this. First, we show that it crucially needs polymatroidal structure of the allocation set and present an instance where the statement not hold for a non-polymatroid allocation set, even when this set has a $1$-majorized point. Further,  the result $\mathcal{B}(g) \subseteq \mathcal{U}$ is not a general result for any convex polytope $\mathcal{U}$ whose saturation function $g$ is monotone and submodular  and whose coordinates sum to a fixed value. Here, for any $S$, the saturation function $g(S)$ is the maximum over $\mathcal{U}$ of the sum of the coordinates in $S$. Our proof uses specific properties of the way $\mathcal{U}$ is defined in our setting. We present an example of a constant coordinate-sum convex polytope $\mathcal{U}$ whose saturation function $g$ is monotone and submodular, but that strictly lies inside $\mathcal{B}(g)$.
\subsection{Polynomial Time  Approximation Scheme}\label{sec:poly-time}

We will next show a polynomial time additive approximation to compute the $1$-majorized point.

\begin{theorem}\label{thm:poly-alg}
    In the full information revelation setting, we can compute a policy that approximates the $1$-majorized vector of utilities to an additive $O(\delta)$ in time \mbox{poly}$(n,V,1/\delta)$.
\end{theorem}


The above theorem also shows that the approximation ratio can easily be made multiplicative $(1+\delta)$ if the running time is $\mbox{poly}(n, V/\delta, \OPT_n/\OPT_1)$, where $\OPT_1$ is the max-min fair utility value and $\OPT_n$ is the social welfare. 

\paragraph{Proof of \cref{thm:poly-alg}.} In the rest of the section, we provide a proof sketch of \cref{thm:poly-alg}. Assume that all values are normalized so that the smallest value is $1$ and the largest value is $V$. The algorithm is an application of the multiplicative weights update framework. 
We use this to compute the maximum sum of each of the smallest $k$ utilities by binary sum, and then use the framework again with these values to compute a sequence of policies that can be computed efficiently by sampling, whose average approximates the majorized solution to an additive $\delta$. Since the details of our approach are very similar to that in~\citet{bhalgat_optimal_2013}, we only present a sketch and omit the details. 



\paragraph{LP for the Optimal Prefix Sum of Utilities.}
    Given a vector $\Vec{x}=(x_1,\dots, x_n)$, let the $i$-th smallest element of $x$ be $x_{(i)}$. We define $Q_j(\Vec{x})=\sum_{i=1}^j x_{(i)}$. Recall that $\Vec{x}$ is majorized by $\Vec{y}$ or $\Vec{y}$ $\alpha$-majorizes $\Vec{x}$ if $\alpha\cdot Q_j(\Vec{y})\geq Q_j(\Vec{x})$.

Let $\Vec{v}=\{v_{a_i}\}_{i=1}^n$ be the realized values of the agents. Using a result in~\citet{goel_simultaneous_2006}, for every $1\leq j\leq n$, the program below finds $\max\{Q_j(\{U_i(\Omega)\}): \Omega\text{ feasible policy}\}$: 
\begin{gather*}
    \text{Maximize }\left(\sum_{i=1}^n U_i'\right)-(n-j)M\text{ subject to:}\\
     U_i\leq \E_{\Vec{v}}\bracket{\sum_{i=1}^m v_{a_i} x_{i\Vec{v}}}, \text{ for all } i\\
     \{v_{a_i} x_{i\Vec{v}}\}\in \mathcal{F}(\Vec{v}), \text{ for all } \Vec{v}\\
     U_i'\leq \min\{U_i,M\}, \text{ for all } i.
\end{gather*}
Here, $\mathcal{F}(\Vec{v})$ is the base polytope $\mathcal{B}(g)$ from the proof of \cref{lem:mat-poly} when the realized value vector is $\Vec{v}$. Let $\OPT^*_j$ denote the optimal objective to the above program.

\begin{lemma}[Lemma 3.1 in \citet{goel_simultaneous_2006}]\label{lem:lp}
    The linear program above finds $\max\{Q_j(\{U_i(\Omega)\}): \Omega\text{ feasible policy}\}$. 
\end{lemma}

\paragraph{Dual Oracle and Multiplicative Weights.}  We now follow the framework in~\citet{bhalgat_optimal_2013} and use the Multiplicative Weights method to decide feasibility of the first constraint subject to all the others. 

For a fixed guess objective value $\OPT_j$ (which we can find the optimal value of via binary search), we rewrite the above LP as a feasibility problem for the objective being at least $\OPT_j$. It suffices to solve the corresponding oracle problem with nonnegative dual multipliers $\{\lambda_i\}$: 
\begin{gather*}
    \text{Maximize }-\sum_i \lambda_i U_i+\E_{\Vec{v}}\bracket{\sum_i \lambda_i v_{a_i} x_{i\Vec{v}}}\\
     \{v_{a_i} x_{i\Vec{v}}\}\in \mathcal{F}(\Vec{v}), \text{ for all } \Vec{v}\\
     U_i'\leq \min\{U_i, M\}, \text{ for all } i\\
     \left(\sum_{i=1}^n U_i' \right)-(n-j)M\geq \OPT_j.
\end{gather*}
This optimization problem decouples into two separate optimization programs. Minimizing $\sum_i \lambda_iU_i$ subject to all but the first constraint is a linear program and can be solved in polynomial time. Finding the maximum of $\E_{\Vec{v}}\bracket{\sum_i \lambda_i v_{a_i} x_{i\Vec{v}}}$ splits into finding the maximum of $\sum_i \lambda_i v_{a_i} x_{i\Vec{v}}$ subject to $\{v_{a_i} x_{i\Vec{v}}\}\in \mathcal{F}(\Vec{v})$ for each $\Vec{v}$.  The work of~\citet{bhalgat_optimal_2013} shows that for parameter $\delta > 0$, we can find a solution to the original LP with value at least $\OPT_j^*$ satisfying 
$$ U_i\leq \E_{\Vec{v}}\bracket{\sum_{i=1}^m v_{a_i} x_{i\Vec{v}}} + \frac{\delta}{n}$$
in time  $\mbox{poly}(n,V, 1/\delta)$, provided the following three conditions hold:

\begin{description}
\item[Sampling:] The quantity $\E_{\Vec{v}}\bracket{\sum_{i=1}^m v_{a_i} x_{i\Vec{v}}}$ (for any given solution $\vec{x}$) can be approximated to $\pm \delta/n$ with high probability by the sample average of $\mbox{poly}(n,V,1/\delta)$ vectors $\Vec{v}$. This follows by Hoeffding's inequality since the range of $\sum_{i=1}^m v_{a_i} x_{i\Vec{v}}$ is $[0,nV]$.
\item[Dual Oracle Optimization.] For non-negative dual multipliers $\{\lambda_i\}_{i=1}^n$, given a vector $\Vec{v}$, there is a polynomial time algorithm that maximizes $\sum_i \lambda_i v_{a_i} x_{i\Vec{v}}$ subject to $\{v_{a_i} x_{i\Vec{v}}\}\in \mathcal{F}(\Vec{v})$. In our case, this is simply the greedy algorithm for optimizing a linear function over the base polytope given by the submodular function $g$ from \cref{lem:mat-poly}, corresponding to the vector $\Vec{v}$. It is easy to check that the greedy algorithm runs in polynomial time, since for any given $S$, $g(S)$ is efficiently computable via the proof of \cref{lem:mat-poly}. 
\item[Width.] The width of the polytope, given by $\max_i | \E_{\Vec{v}}\bracket{\sum_{i=1}^m v_{a_i} x_{i\Vec{v}}} - U_i|$ over solutions feasible to the other constraints is polynomially bounded. In our case, the width is bounded by $O(n V)$. 
\end{description}

Binary searching over $\OPT_j$ to find the largest value that ensures feasibility of the program, this implies a feasible solution to the original LP whose objective (prefix sum) is at least $\OPT^*_j - \delta$. Running this procedure separately for each $1\leq j\leq n$, we find $\{\OPT_j\}_{j=1}^n$, such that for each $j$, \[\OPT_j\geq \max\{Q_j(\{U_i(\Omega)\}): \Omega\text{ feasible policy}\}-\delta.\]




\paragraph{Computing $\OPT_j^*$ and the Feasible Policy.} 
Now that we have computed the final $\OPT_j$, we still need to compute the feasible policy that guarantees these prefix sums of utilities simultaneously for all $j$. For this, we combine all the above LPs into a single one. Since each finally computed $\OPT_j$ is guaranteed to be close to $\max\{Q_j(\{U_i(\Omega)\}): \Omega\text{ feasible policy}\}$, the following LP is feasible: 
\begin{gather*}
     U_i\leq \E_{\Vec{v}}\bracket{\sum_{i=1}^m v_{a_i} x_{i\Vec{v}}}, \text{ for all } 1\leq i\leq n\\
     \{v_{a_i} x_{i\Vec{v}}\}\in \mathcal{F}(\Vec{v}), \text{ for all } \Vec{v}\\
     U_{ij}'\leq \min\{U_i, M_j\}, \text{ for all } 1\leq i\leq n,1\leq j\leq n\\
     \left(\sum_{i=1}^n U_{ij}'\right)-(n-j)M_j\geq \OPT_j, \text{ for all }1\leq j\leq n.
\end{gather*}
A similar application of the multiplicative weights update method to this program (taking the Lagrangian of the first set of constraints) now computes the final feasible policy assuming the utilities are approximated by an additive $O(\delta)$. For the policy itself, we sample a random time step and consider the dual variables $\{\lambda_i\}$ output by the procedure. Given a vector of revealed values $\Vec{v}$, we simply maximize $\sum_{i=1}^n \lambda_i v_{a_i} x_{i\Vec{v}}$ subject to $\{v_i x_{i\Vec{v}}\} \in \mathcal{F}(\Vec{v})$ to find the allocation rule. The details are similar to~\citet{bhalgat_optimal_2013}. This completes the sketch of the proof of \cref{thm:poly-alg}.

\section{General Signaling Policies and Approximate Majorization}\label{sec:approx}
 We now build on the results in \cref{sec:full-rev} to show \cref{thm:main-informal-exist,thm:main-informal-compute}, the existence of approximate majorized policies, and associated computational result, for general policies. To appreciate the technical challenge, the mapping scheme in full revelation policies is fixed, so that we only need to focus on designing the selection policy (or allocation rule) of the receiver. This makes the overall problem have polymatroid structure if it has that structure for a fixed scenario. However, for general signaling policies, there is a dependence between the mapping rule in the signaling policy and the allocation made by the receiver. Since both the mapping rule and selection rule are not fixed anymore, the overall signaling problem may not have polymatroid structure. 

To extend our result to approximate majorization of general policies, we adopt the approach of randomized single mean projections introduced in \citet{banerjee_majorized_2024} for selecting a single agent. For these policies, it was shown that there is a fixed set of mappings, termed maximal mappings, that can be pre-computed and are optimal (in terms of majorization) within this class. These mappings allowed them to approximate any mapping policy by an analog of full revelation policies. We follow this outline; however, we need a different set of technical arguments to show that this class of policies suffice. The main novelty in our case, beyond extending \cref{lem:mat-poly} to single-mean policies, is the proof of \cref{lem:maximal-mapping-dominance}, which carefully uses submodularity to show that it suffices to consider maximal mappings.

\subsection{Single Mean Projections}
\label{sec:single-mean}
The definitions in this section mirrors that in \citet{banerjee_majorized_2024}. We briefly review the definitions for completeness. Recall that the values $v_i$ of the agents are supported on $[1,V]$. Intuitively, single mean projection partitions the value range $[1,V]$ into a sequence of buckets, and only counts utility from one bucket. 
Given small $\eps>0$, let $\eta=1+\eps$. Assume $V$ is a power of $\eta$. Divide $[1,V]$ into buckets $I_1=[1,\eta), I_2=[\eta,\eta^2),\dots, I_k=[V/\eta,V)$.   Let $K = O\left(\frac{\log V}{\eta}\right)$ denote the number of buckets. We will use these buckets $\{I_k\}_{k=1}^K$ to partition the range of posterior means, where each bucket $I_k$ is associated with a canonical mean value $m_k$. We will construct signaling policies that choose a bucket at random and focus on the case where the posterior mean lies within this bucket.

\paragraph{Approximate Welfare-maximizing Receiver.}
 As discussed in \cref{sec:intro}, a key element of our approach to generalizing to arbitrary signals is to model the receiver not as a perfect optimizer over the exact posterior means, but as an approximate one who acts on canonical values. This models a receiver who has bounded rationality. We now formalize this model.

Given a vector of posterior means $\mu = (\mu_1, \dots, \mu_n)$, the receiver behaves as follows:
\begin{enumerate}
    \item For each agent $i$, the receiver identifies the bucket $I_{k_i}$ such that $\mu_i \in I_{k_i}$.
    \item The receiver constructs a canonical value vector $\mu'$ by setting $\mu'_i = m_{k_i}$ for all $i \in E$.
    \item The receiver computes an allocation vector $x(\mu')$ by running the greedy algorithm on the polymatroid $\mathcal{P}(f)$ with the canonical value vector $\mu'$ as the objective.
\end{enumerate}

This models a receiver who is a $(1+\epsilon)$ approximate welfare maximizer in each dimension, acting on canonical values rather than exact posterior means. From now on, we will ignore the $(1+\epsilon)$ factor, and assume the utility of an agent is computed using the canonical posterior mean values.

\paragraph{Single-Mean Projections.} 
Suppose the receiver is an approximate welfare maximizer as defined above. We define a single-mean policy for bucket $I_k$ as follows:

\begin{definition} [Single-mean Policies and Fake Utilities]
\label{def:fake}
Consider any signaling policy specified by a mapping rule and a selection rule. For any bucket $I_k$, the corresponding single-mean policy restricted to that bucket accounts for the utility of any agent as follows. The \emph{fake utility} of an agent $i$, denoted by $\hU_{i, k}$, is measures as the utility of agent $i$ when the posterior mean lies in $I_k$, else zero. Formally, if $\vec{x}$ denotes the allocation, then
\[
\hU_{i,k}(\vec{x}) = 
\begin{cases}
    m_k \cdot x_i, & \text{if $\mu_i \in I_k$;}\\
    0, & \text{if $\mu_i \notin I_k$}.
\end{cases}
\]
\end{definition}

Note that the only difference in a single-mean policy and a regular policy is the utility accounting as a fake utility. This fake utility serves as an underestimation of the true utility seen by the agent in the policy.  We now define a randomized single-mean policy as follows:

\begin{definition} [Randomized Single-mean Policies]
    A randomized single-mean policy is constructed as follows: The mapping rule $\Delta$ is a collection of $K$ mapping rules $\Delta_1, \ldots, \Delta_K$ and associated selection rules, yielding a collection of signaling policies $\Omega_1, \ldots, \Omega_K$. The utility of policy $\Omega_k$ is measured using the fake utility restricted to bucket $I_k$. The overall policy chooses one of the $K$ buckets uniformly at random, and uses the corresponding signaling policy $\Omega_k$. 
\end{definition}

Note that the expected fake utility of agent $i$ in the above randomized single-mean policy is
$$ \hU_i(\Omega) = \frac{1}{K} \sum_{k=1}^K \hU_{i,k}(\Omega_k).$$

Note that given any signaling policy $\Omega$, there is a randomized single-mean policy $\Omega_{\text{rsm}}$ obtained by picking a bucket uniformly at random and using the fake utility restricted to that bucket. In this policy, we have 
$  \hat{U}_i(\Omega_{\text{rsm}}) =   \frac{U_i(\Omega)}{K}$ for all agents $i$,
where $\hat{U}_i(\Omega_{\text{rsm}})$ is the expected fake utility of agent $i$ in the randomized single mean policy. Note that we have accounted for $U_i(\Omega)$ by rounding each posterior mean to its canonical value, and ignored the $(1+\epsilon)$ factor loss in this process. 

We now consider fixing the mapping rule $\Delta$ used by the randomized single-mean policies, but do not fix the selection rule used by the receiver. In other words, for every  bucket $I_k$, we specify the mapping rule $\Delta_k$ of values to signals. Note that this bucket is chosen with probability $1/K$ in the rule $\Delta$.  The following is analogous to \cref{lem:mat-poly}. 

\begin{lemma}
\label{lem:randomized-bucket-submodularity}
Fix an active bucket $I_k$, and the mapping rule $\Delta_k$ of the corresponding single mean policy. Fix the vector of posterior means $\mu$ of the agents. Let $\hat{g}_k(S;\mu)$ denote the maximum possible sum of the fake utilities of agents in a set $S$, where the maximization is over the selection rule of the receiver.  In other words,
$$ \hat{g}_k(S;\mu) = \max \sum_{i\in S} \hat{U}_{i,k} $$
where $\hat{U}_{i,k}$ is as defined in \cref{def:fake}. Then, the function $\hat{g}_k(S;\mu)$ is monotone and submodular.
\end{lemma}
\begin{proof}
Given $\mu$, the receiver runs a greedy allocation that sorts agents in decreasing order of their canonical posterior means. By the argument in the proof of \cref{lem:mat-poly}, the maximum value is achieved by running the greedy algorithm on the canonical posterior means, with a tie-breaking rule that prioritizes agents in the set $S$. Let $\pi_S$ be the permutation corresponding to this rule. 
Let the blocks of agents with identical canonical posterior means be denoted by $\{E'_j\}_{j=1}^K$, where agents in $E'_j$ have canonical mean $m_j$, meaning $\mu_i \in I_j$. The bucket of interest $I_k$ corresponds to one of these blocks, say $E'_k$, with canonical value $m_k$. Then, 
\[ \hat{g}_k(S;\mu) = \max_{\vec{x}} \sum_{i \in S \cap E'_k} m_k \cdot x_i = m_k \sum_{i \in S \cap E'_k} x_i(\pi_S). \]
Let $h_k(S;\mu) = \sum_{i \in S \cap E'_k} x_i(\pi_S)$.  Let $P_{>k} = \cup_{j > k} E'_j$ be the fixed set of all agents in higher-value blocks. The greedy algorithm (given by ordering $\pi_S$) processes all agents in the set $S_k = S \cap E'_k$ contiguously and before other agents in $E'_k$. The sum of the allocations for these agents forms a telescoping series:
\[ h_k(S;\mu) = \sum_{i \in S_k} x_i(\pi_S) = f(P_{>k} \cup S_k) - f(P_{>k}). \]
where $f$ is the rank function of the underlying polymatroid constraint. This implies
$$ \hat{g}_k(S;\mu) =  m_k \cdot \left( f(P_{>k} \cup \left(S \cap E'_k\right)) - f(P_{>k}) \right). $$
This function is clearly submodular, completing the proof.
\end{proof}

We now proceed as in the proof of \cref{lem:mat-poly} and how the following lemma:
\begin{lemma}
    Fix an active bucket $I_k$, and the mapping rule $\Delta_k$ of the corresponding single mean policy. Fix the vector of posterior means $\mu$ of the agents. Let $\mathcal{U}$ denote the set of achievable fake utility vectors $\{\hU_{i,k}\}_{i=1}^n$. Then 
    $$ \mathcal{B}(\hat{g}_k(\cdot;\mu))  \subseteq \mathcal{U} \subseteq \mathcal{P}(\hat{g}_k(\cdot;\mu)),$$
    where $\hat{g}_k(\cdot;\mu)$ is the set function defined in \cref{lem:randomized-bucket-submodularity}.
\end{lemma}
\begin{proof}
By the definition of $\hat{g}_k(\cdot;\mu)$, for any set $S \subseteq [n]$, we have $\sum_{i \in S} \hU_{i,k} \le \hat{g}_k(S;\mu)$. The latter system defined $ \mathcal{P}(\hat{g}_k(\cdot;\mu))$, which shows the second containment. 

We will now show the first containment. First, note that $\mathcal{U}$ is the projection of the true utility vector (computed using canonical means) onto the coordinates in $E'_k$ (using the notation from the proof of \cref{lem:randomized-bucket-submodularity}). The set of true utility vectors  is derived from the face of the polymatroid $\mathcal{P}(f)$ that maximizes $\sum_{i=1}^n \hat{\mu}_i x_i$, where $\hat{\mu_i} = m_j$ if $i \in E'_j$ --- each optimal point $\vec{x}$ yields the utility vector $\{\hat{\mu_i} \cdot x_i\}_{i=1}^n$. This set is therefore convex, and its projection $\mathcal{U}$ is convex as well. 

Consider the base polytope $\mathcal{B}(\hat{g}_k(\cdot;\mu))$. By \cref{lem:greedy2}, any vertex of the base polytope corresponds to a permutation $\pi$ of the agents (say $1,2,\ldots,n$) and setting $\hat{U}_{i,k}= \hat{g}_k([i];\mu) - \hat{g}_k([i-1];\mu)$.  We will show this vector is realizable by some tie-breaking rule of the receiver's greedy algorithm. First, if $i \notin E'_k$, then $\hat{U}_{i,k} = 0$ by the formula derived in \cref{lem:randomized-bucket-submodularity}. Next, suppose $i \in E'_k$, and let $S_{i,k} = [i-1] \cap E'_k$, then by the same formula, if we place $i$ after $S_{i,k}$ in the receiver's tie-break ordering of the agents in $E'_k$, then the receiver allocates $x_i = f(P_{>k} \cup S_k \cup \{i\}) - f(P_{>k} \cup S_k)$  in its greedy algorithm. This means the fake utility of $i$ in the receiver's allocation  is exactly $\hat{U}_{i,k}$, showing $\{\hat{U}_{i,k}\}_{i=1}^n$ is realizable.  Thus, the vertices of $\mathcal{B}(\hat{g}_k(\cdot;\mu))$ are contained in $\mathcal{U}$, and by the convexity of both sets, we have $\mathcal{B}(\hat{g}_k(\cdot;\mu))  \subseteq \mathcal{U}$, completing the proof. 
\end{proof}

Taking the Minkowski sum over the random choice of bucket $I_k$ and  over the realized posterior means, and using \cref{cor:mon-sub}, we obtain the following.

\begin{corollary}
\label{cor:base2}
    Given a randomized single mean mapping rule $\Delta = \{\Delta_k\}_{k=1}^K$, let the expected fake utility be $\hU_i = \frac{1}{K} \sum_{k=1}^k \E_{\mu \sim \Delta_k}[\hU_{i,k}] $. Let $\mathcal{U}$ denote set of vectors $\{\hU_i\}_{i=1}^n$ obtained by varying the selection rule of the receiver. Then, $\mathcal{U}$ is contained in the following polymatroid $\mathcal{\hat{R}}(\Delta)$, and contains its base polytope, where $\mathcal{\hat{R}}(\Delta) $ is defined as: 
\begin{equation}
\label{eq:poly2}
\mathcal{\hat{R}}(\Delta) = \left\{ \vec{y} \ge 0\ \big| \ \sum_{i \in S} y_i \le \frac{1}{K} \sum_{k=1}^K \E_{\mu \sim \Delta_k} [\hat{g}_k(S;\mu)] \ \forall S \subseteq [n] \right\}. 
\end{equation}
Here, the expectation in the RHS is over the vector $\vec{\mu}$ of posterior means produced by the mapping $\Delta_k$, and the functions $\hat{g}_k$ are as defined in \cref{lem:randomized-bucket-submodularity}. 
\end{corollary}


\subsection{Maximal Single-mean Mappings}
\label{sec:max-map}
The issue with extending the above lemma to a proof analogous to \cref{thm:maj-policy} is that the mapping rule $\Delta$  is now a variable (in addition to the selection rule that depends on $\Delta$). This wasn't an issue in the proof of \cref{thm:maj-policy}, where the mapping rule was fixed and the set of utilities obtained by varying the selection rule defines the base of a polymatroid. In contrast, though $\mathcal{\hat{R}}(\Delta)$ as  defined above is a polymatroid, the union of such polymatroids over $\Delta$ need not have nice structure.

We now proceed as in~\citet{banerjee_majorized_2024} and show that the optimal signaling policy for single-mean policies is fixed and independent of the allocation. This will allow us to argue polymatroidal structure, and show that the space of randomized single mean policies has a $1$-majorized solution. Since the true utility is within a factor of $K$ of the fake utilities used by such policies, this directly implies a $K$-majorized policy for general signaling policies, and we will show that in \cref{thm:main-formal}.

Towards this end, we define a {\em maximal mapping} analogous to~\citet{banerjee_majorized_2024}.

\begin{definition}[Maximal Mapping] 
   For an interval $I_k$, a \emph{maximal mapping} is a mapping rule $\omega$ from agent values to signals $\{\sigma\}$ such that  $\Pr_{\sigma \sim \omega} [\mu_i(\sigma) \in I_k]$ is maximized for each agent $i$.
\end{definition}

Note that for each agent $i$, the maximal mapping to a given interval $I_k$ is the solution to a linear program~\citep{banerjee_majorized_2024}. This mapping is fixed and decoupled from the allocation rule of the receiver. It can also be computed separately for each agent. The set of maximal mappings, one for each $I_k$, yields the mapping rule of a randomized single-mean policy, by choosing one of the buckets uniformly at random. We call this mapping rule $\Delta_{\max}$.

We now present the key structural lemma to show that randomized single mean policies can switch to using maximal mappings without reducing the maximum expected utility of any set of agents. In the lemma below, the notation $\hat{g}_k(S;\mu)$ is as defined in the proof of \cref{lem:randomized-bucket-submodularity}. Further, by the notation $\mu \sim \Delta_k$, we mean a posterior mean vector that results from the execution of the mapping rule $\Delta_k$. 

\begin{lemma}[Structure Lemma]
\label{lem:maximal-mapping-dominance}
Consider a fixed active bucket $I_k$. Let $\Delta_k = (\sigma_1, \dots, \sigma_n)$ be an arbitrary mapping rule where each agent's mapping $\sigma_i$ is chosen independently. Consider the mapping rule $\Delta^{\max}_{k} = (\sigma'_1, \dots, \sigma'_n)$, where each $\sigma'_i$ is a maximal mapping for agent $i$ with respect to the bucket $I_k$. Then, for any set $S$ of agents, we have:
\[ \mathbb{E}_{\mu \sim \Delta^{\max}_{k}}[\hat{g}_k(S; \mu)] \ge \mathbb{E}_{\mu \sim \Delta_k}[\hat{g}_k(S; \mu)]. \]
\end{lemma}
\begin{proof}
The proof proceeds by showing that for any single agent $i$, changing its signaling policy from an arbitrary one, $\sigma_i$, to its maximal mapping, $\sigma'_i$, while keeping all other agents' schemes fixed, can only increase the total expected utility. Since the agents' signaling policies are independent, iterating this argument over all agents will complete the proof.

We fix one agent, say agent 1, and consider changing its mapping from $\sigma_1$ to $\sigma'_1$. Let $\Delta^{-1} = (\sigma_2, \dots, \sigma_n)$ be the policies for all other agents. Let $\mu_{-1} = (\mu_2, \dots, \mu_n)$ be a realization of posterior means for these agents, drawn according to $\Delta^{-1}$.

For a given realization $\mu_{-1}$, we define a function $H_{\mu_{-1}}(\mu_1)$ as 
\[ H_{\mu_{-1}}(\mu_1) = \hat{g}_k(S; (\mu_1, \mu_{-1})). \]
The total expected utility can be written as an expectation over the choices of all agents:
\[ \mathbb{E}_{\mu \sim \Delta_k}[\hat{g}_k(S; \mu)] = \mathbb{E}_{\mu_{-1} \sim \Delta^{-1}} \left[ \mathbb{E}_{\mu_1 \sim \sigma_1}[H_{\mu_{-1}}(\mu_1)] \right]. \]
It therefore suffices to show that for any fixed outcome $\mu_{-1}$ of the other agents, the inner expectation increases when we switch from $\sigma_1$ to $\sigma'_1$:
\[ \mathbb{E}_{\mu_1 \sim \sigma'_1}[H_{\mu_{-1}}(\mu_1)] \ge \mathbb{E}_{\mu_1 \sim \sigma_1}[H_{\mu_{-1}}(\mu_1)]. \]

The value of $H_{\mu_{-1}}(\mu_1)$ depends on which bucket, $I_j$, the value $\mu_1$ falls into. Let $E_j(\mu_{-1}) = \{i \ne 1 \mid \mu_i \in I_j\}$ be the set of other agents in bucket $j$.
From our previous analysis, the closed-form for the utility is:
\[ \hat{g}_k(S; \mu) = m_k \left[ f(P_{>k}(\mu) \cup S_k) - f(P_{>k}(\mu)) \right]. \]
where $S_k$ is the set of agents falling within bucket $I_k$ and $P_{>k}(\mu)$ is the set of agents falling in a bucket with higher mean than bucket $k$. (These are random variables since we did not fix the bucket of agent $1$.) 

We will now analyze how the bucket assignment of agent $1$ affects this value, for a fixed $\mu_{-1}$. Let $P_{>k}^{-1} = \bigcup_{j>k} E_j(\mu_{-1})$ and $S_k^{-1} = S \cap E_k(\mu_{-1})$. These are fixed values since we fixed $\mu_{-1}$. Then, we have:

\begin{itemize}
    \item If $\mu_1 \in I_k$, then
    $P_{>k}(\mu) = P_{>k}^{-1}$ and $S_k = S_k^{-1} \cup (S \cap \{1\})$.
    The utility expression becomes:
    \[ H_{\mu_{-1}}(\mu_1 \in I_k) = m_k \left[ f(P_{>k}^{-1} \cup S_k^{-1} \cup (S \cap \{1\})) - f(P_{>k}^{-1}) \right]. \]

    \item If $\mu_1 \in I_j$ where $j > k$, then
    $P_{>k}(\mu) = P_{>k}^{-1} \cup \{1\}$, and $S_k = S_k^{-1}$.
    Then:
    \[ H_{\mu_{-1}}(\mu_1 \in I_j, j>k) = m_k \left[ f(P_{>k}^{-1} \cup \{1\} \cup S_k^{-1}) - f(P_{>k}^{-1} \cup \{1\}) \right]. \]
    
    \item If $\mu_1 \in I_j$ where $j < k$, then
    $P_{>k}(\mu) = P_{>k}^{-1}$ and $S_k = S_k^{-1}$. 
    The utility expression becomes:
    \[ H_{\mu_{-1}}(\mu_1 \in I_j, j<k) = m_k \left[ f(P_{>k}^{-1} \cup S_k^{-1}) - f(P_{>k}^{-1}) \right]. \]
\end{itemize}

Let $C_{in} = H_{\mu_{-1}}(\mu_1 \in I_k)$, $C_{above} = H_{\mu_{-1}}(\mu_1 \in I_j, j>k)$, and $C_{below} = H_{\mu_{-1}}(\mu_1 \in I_j, j<k)$. By monotonicity of $f$, all these terms are non-negative and we have $C_{in} \ge C_{below}$. By the submodularity of $f$, we know that for any sets $X, Y$ and any element $t \notin X$, we have $f(X \cup Y) - f(X) \ge f(X \cup \{t\} \cup Y) - f(X \cup \{t\})$. Applying this with $X=P_{>k}^{-1}$, $Y=S_k^{-1} \cup (S \cap \{1\})$, and $t=1$, we have $C_{in} \ge C_{above}$.  Therefore, we have $C_{in} \ge \max\{C_{above}, C_{below}\}$.

Let $p_j = \Pr_{\mu_1 \sim \sigma_1}[\mu_1 \in I_j]$ be the probabilities under the original scheme, and $p'_j$ be the probabilities under the maximal mapping $\sigma'_1$. The expected utilities for given $\mu_{-1}$ are:
\begin{align*}
    \mathbb{E}_{\mu_1 \sim \sigma_1}[H_{\mu_{-1}}(\mu_1)] &= p_k C_{in} + \sum_{j>k} p_j C_{above} + \sum_{j<k} p_j C_{below}, \\
   \mathbb{E}_{\mu_1 \sim \sigma'_1}[H_{\mu_{-1}}(\mu_1)] &= p'_k C_{in} + \sum_{j>k} p'_j C_{above} + \sum_{j<k} p'_j C_{below}.
\end{align*}

The construction in the proof of Lemma 4.3 of ~\citet{banerjee_majorized_2024} shows that given any mapping $\sigma$, one can construct a new signal that moves probability mass from signals whose posteriors lie in the sets $\{I_j, j > k\}$ and $\{I_j, j < k\}$ (assuming both masses are non-zero) to signals whose posterior lie in the set $I_k$, while preserving Bayes plausibility, meaning that the posteriors correspond to a valid signaling policy and the expectation of the posterior means is equal to the prior mean. We present the pooling process from~\citet{banerjee_majorized_2024} below for the sake of completeness:

\begin{itemize}
    \item Consider the mapping $\sigma$ and let $S_1$ be the set of signals with posterior means $\{\mu_1 \in I_j, j < k\}$, $S_2$ be the signals with $\{\mu_1 \in I_k\}$ and $S_3$ be the set of signals $\{\mu_1 \in I_j, j > k\}$. Construct signals $\phi^1, \phi^2, \phi^3$ and send them whenever a signal in $S_1, S_2, S_3$ is sent, respectively. This step preserves $p_k, \sum_{j > k} p_j, \sum_{j < k} p_j$.
    \item Let $\eta_1, \eta_2, \eta_3$ denote the posterior means of the signals $\phi_1, \phi_2, \phi_3$ respectively. Note that $\eta_1 \in I_j$ for some $j < k$; $\eta_2 \in I_k$, and $\eta_3 \in I_j$ for some $j > k$. Let $\eta_2 = \alpha \eta_1 + (1-\alpha)\eta_3$, where $\alpha \in (0,1)$. define $q_1 := \Pr[\phi^1]$ and $q_3 := \Pr[\phi^3]$. Create a new signal $\phi^4$, and do the following things:
    \begin{itemize}
        \item If $\beta = \frac{q_1 (1-\alpha)}{q_3 \alpha} \le 1$, then whenever signal $\phi^1$ was sent,  $\phi^4$ is sent instead, and whenever $\phi^3$ was sent, $\phi^4$ is sent instead with probability $\beta$ and $\phi^3$ is sent with probability $1-\beta$. 
        \item If $\beta > 1$, then whenever signal $\phi^3$ was sent, $\phi^4$ is sent instead, and whenever $\phi^1$ was sent, $\phi^4$ is sent instead with probability $1/\beta$ and $\phi^1$ sent with probability $1-1/\beta$.
    \end{itemize}
    \item In either case, we note that the posterior mean of $\phi^4$ is precisely $\alpha \eta_1 + (1-\alpha_i) \eta_3 = \eta_2 \in I_k$. The sender can then send signal $\phi^2$ whenever $\phi^4$ is sent.
\end{itemize}

This pooling process does not increase $\sum_{j < k} p_j = \Pr[\phi^1]$ and $\sum_{j > k} p_j = \Pr[\phi^3]$, while it does not decrease $p_k = \Pr[\phi^2]$, which is the probability of the posterior mean landing in the target bucket $I_k$. It also does not decrease $\mathbb{E}_{\mu_1 \sim \sigma_1}[H_{\mu_{-1}}(\mu_1)]$, since $C_{in} \ge \max\{C_{below}, C_{above}\}$. Note that now, $\min\{\sum_{j>k} p_j, \sum_{j<k} p_j\}  = 0$. Since $p_k$ did not decrease in this process, a maximal mapping $\sigma_1'$ will already be of this form, that is, we also have $\min\{\sum_{j>k} p'_j, \sum_{j<k} p'_j\} = 0$.

Let $\theta = \E[D_1]$ denote the prior mean of agent $1$. There are two cases:

\begin{description}
\item[Case 1: $\theta \in I_k$:] In this case, the maximal mapping satisfies $p'_k = 1$, so that $\max\{\sum_{j>k} p'_j, \sum_{j<k} p'_j\} = 0$. Since $C_{in} \ge \max\{C_{below}, C_{above}\}$ and the probabilities $p$ and $p'$ respectively sum to $1$, this implies $\mathbb{E}_{\mu_1 \sim \sigma'_1}[H_{\mu_{-1}}(\mu_1)] \ge \mathbb{E}_{\mu_1 \sim \sigma_1}[H_{\mu_{-1}}(\mu_1)]$.
\item[Case 2: $\theta \notin I_k$.] Suppose $\theta \in I_j$ for $j < k$. Then, Bayes plausibility (posterior mean equals prior mean) implies $\sum_{j<k} p_j > 0$ and $\sum_{j<k} p'_j > 0$. Since $p'_k \ge p_k$ (by the maximal mapping property), this implies $\sum_{j<k} p'_j \le \sum_{j<k} p_j$. This again implies $\mathbb{E}_{\mu_1 \sim \sigma'_1}[H_{\mu_{-1}}(\mu_1)] \ge \mathbb{E}_{\mu_1 \sim \sigma_1}[H_{\mu_{-1}}(\mu_1)]$. A similar argument holds when  $\theta \in I_j$ for $j > k$.
\end{description}

We therefore have $\mathbb{E}_{\mu_1 \sim \sigma'_1}[H_{\mu_{-1}}(\mu_1)] \ge \mathbb{E}_{\mu_1 \sim \sigma_1}[H_{\mu_{-1}}(\mu_1)]$.  Since this holds for any $\mu_{-1}$, it also holds after taking the expectation over $\mu_{-1} \sim \Delta^{-1}$. By iterating this argument for all agents, we conclude that using a maximal mapping for every agent is optimal for maximizing $\E_{\mu \sim \Delta_k} [\hat{g}_k(S;\mu)]$.
\end{proof}

The above lemma implies the following corollary:

\begin{corollary}
\label{cor:maximal}
Consider any randomized single mean policy $\Omega$ and let $\vec{U}$ denote the vector of expected (fake) utilities of the agents in this policy. Then  $\vec{U} \in \mathcal{\hat{R}}(\Delta_{\max})$, where $\Delta_{\max}$ is the mapping obtained by choosing one of the $K$ buckets uniformly at random and using the corresponding maximal mapping for each agent.
\end{corollary}
\begin{proof}
    The vector $\vec{U}$ belongs to  the polymatroid $\mathcal{\hat{R}}(\Delta)$ in \cref{eq:poly2}, where the mapping $\Delta$ corresponds to the policy $\Omega$.  By \cref{lem:maximal-mapping-dominance}, if we use the maximal mapping in each bucket instead, the RHS of the constraints in \cref{eq:poly2} do not decrease. This means $\vec{U} \in \mathcal{\hat{R}}(\Delta_{\max})$. 
\end{proof}


\subsection{Main Result: Proof of \cref{thm:main-informal-exist,thm:main-informal-compute}}
We now combine the structural results from the preceding sections to formally state and prove our main theorem, which establishes the existence of a computationally efficient signaling policy with a logarithmic approximation guarantee for majorization. 
We again note that this result is complemented by a lower bound from~\citet{banerjee_majorized_2024} that rules out an $o(\log \log V)$-majorized policy even for selecting one agent.

\begin{theorem}
\label{thm:main-formal}
Consider the signaling problem with a polymatroid constraint $\mathcal{P}(f)$, where agents have independent quality distributions supported on $[1, V]$. Assume the receiver is a $(1+\varepsilon)$-approximate welfare maximizer who acts on canonical posterior means derived from a partition of $[1, V]$ into $K = O((\log V)/\varepsilon)$ buckets. Then, there exists a signaling policy $\Omega$ that is $O((\log V)/\varepsilon)$-majorized over the set of all possible independent signaling policies. Furthermore, a policy that yields an additive $O(\delta)$ approximation to this utility vector can be computed in time polynomial in $\frac{n V}{\epsilon \cdot \delta}$.
\end{theorem}
\begin{proof}
The proof consists of two parts. First, we prove the existence of a policy with the stated approximation guarantee by relating any optimal policy to the randomized single mean policies described above. Second, we argue that this policy can be computed in polynomial time using the multiplicative weights update framework from \cref{sec:full-rev}.

\paragraph{Proof of \cref{thm:main-informal-exist}.} We now show the existence result. Let $\Omega^*$ be any signaling policy. Let $U(\Omega^*)$ be the vector of true expected utilities for this optimal policy.  As discussed before, consider a randomized single mean policy, $\Omega_{\text{rsm}}$, which is constructed from $\Omega^*$. This policy works by first choosing a bucket $k \in \{1, \dots, K\}$ uniformly at random and creating a fake utility function that only grants the utility an agent would have received from that specific bucket in the original policy $\Omega^*$. The expected utility for agent $i$ under this constructed policy is $U(\Omega_{\text{rsm}}) = U(\Omega^*) / K$.


By \cref{cor:maximal}, the utility vector $\{U(\Omega_{\text{rsm}})\}_{i=1}^n \in \mathcal{\hat{R}}(\Delta_{\max})$, where $\Delta_{\max}$ is the mapping rule that first chooses a bucket $k$ uniformly at random, and then implements the maximal mapping rule for that bucket. Consider the class of signaling policies $\mathcal{C}_{\text{max}}$ that use $\Delta_{\max}$ as their mapping rule. In such policies, the mapping rule is now decoupled from the selection rule since the mapping  $\Delta_{\max}$ can be pre-computed.  By \cref{cor:base2}, the set of expected utility vectors achievable by policies in $\mathcal{C}_{\text{max}}$ lies within the polymatroid $\mathcal{\hat{R}}(\Delta_{\max})$ and contains its base polytope, and this set is non-empty. By \cref{thm:poly-major}, this base polytope has a signaling policy, call it $\Omega_{\text{maj}}$, which is $1$-majorized over all policies in $\mathcal{C}_{\text{max}}$, and hence over all  vectors in $\mathcal{\hat{R}}(\Delta_{\max})$. In particular,   this means $U(\Omega_{\text{maj}})$ majorizes $U(\Omega_{\text{rsm}})$, which is at least $U(\Omega^*) / K$.  This means the utility vector $U(\Omega_{\text{maj}})$ majorizes $U(\Omega^*) / K$.  This proves the existence of a policy that is $O\left(\frac{\log V}{\epsilon}\right)$-majorized, since $K = O\left(\frac{\log V}{\epsilon}\right)$.

\paragraph{Proof of \cref{thm:main-informal-compute}.}  We next sketch the computational result.  We use the multiplicative weights approach from \cref{sec:poly-time}, where we use the posterior (bucketed) mean vector $\mu'$ found by the maximal mapping instead of the value vector. The core requirement for such methods to be efficient is the existence of a polynomial-time oracle for maximizing any linear function over the following polymatroid. Given $\mu'$, the polymatroid has rank function $\hat{g}_k(S; \mu')$. The dual oracle must solve $\max_{u} w \cdot u$ over this polymatroid for a given weight vector $w$. This can be solved by the polymatroid greedy algorithm, which requires oracle access to the rank function $\hat{g}_k(S;\mu')$, followed by sampling scenarios $\mu'$. By combining this with a binary search over the optimal utility values (the prefix sums $Q_j$), we obtain a  polynomial-time additive approximation scheme  for computing the desired logarithmically-approximate majorized policy. We omit the details as they are similar to the proof in \cref{sec:poly-time}.  This concludes the proof of \cref{thm:main-formal}.
\end{proof}
\section{Extensions and Open Questions}
\label{sec:discuss}
Our main contribution is a structural characterization of the utility space in Bayesian persuasion with polymatroid constraints, showing it forms a base polytope of a different polymatroid. This enabled a direct geometric approach to construct a logarithmically-approximate majorized signaling policy. This result highlights a new connection between the geometry of information design and the combinatorial structure of submodular optimization, with potential applications to other information design problems. 

Our techniques easily extend to the setting where the utility of an agent is a fixed multiplier of its allocation, rather than allocation multiplied by the quality (or value). The former case is simpler, since the utility vector now coincides with the allocation vector (appropriately scaled). For a welfare maximizing receiver, the resulting set of utility vectors is trivially a face of $\mathcal{P}(f)$ and is hence the base polytope of a polymatroid. This observation extends the results in the paper to show the same approximation factor for majorization.

Our work leaves several questions open. One open question is to extend our results to the case where the intermediary can correlate the signals between agents. Another question is to understand the combinatorics of the induced fairness polyhedron. Our proof constructs its rank function $g$, but we do not study its interpretation. For instance, if the original constraint is a randomization over independent sets of a matroid, then how is the induced base polytope related to the original matroid? Next, can we design efficient algorithms to find a specific point that maximizes, for instance, the Nash welfare or max-min fairness? Our results imply a logarithmic approximation in polynomial time, but  it is likely these problems admit to a FPTAS. 

At a higher level, it would be interesting to explore other models of allocation. For instance, what if information revelation has a cost, so that, say, the sender's signals are constrained to focus on a few agents? Similarly, what if the agents arrive one at a time, with both the sender and the receiver knowing their priors upfront, while the receiver has to make irrevocable allocations to each arriving agent based on its signal? Finally, can we apply fair persuasion to settings where the receiver is solving a stochastic optimization problem, where for instance, performing two-stage optimization to design a network over the agents, or running a prophet pricing algorithm over the agents~\citep{tang2024intrinsic}. These questions offer a rich domain for structural and algorithmic inquiry. 

\paragraph{Acknowledgment.} We have used Gemini 2.5 Pro to paraphrase and strengthen some text and to perform literature search.

\bibliographystyle{plainnat}
\bibliography{references}
\newpage
\appendix
\section{Example Illustrating Signaling and Fairness}
\label{app:example}

We present an example to demonstrate the simultaneous failure of naive information policies to achieve a good approximation ratio, and the power of a carefully designed signaling policy. We consider $n+1$ agents $\{0, 1, \dots, n\}$, with the constraint that at most one agent can be selected. The polymatroid is therefore the set of allocation vectors (probability of selection) that have $\ell_1$ norm at most one. The receiver selects the agent with highest posterior mean, using a randomized tie-breaking rule specified by the sender. 

Let $q=1/\sqrt{n}$. Agent 0 has a deterministic quality $v_0=2-q$, while agents $i \in \{1,\dots,n\}$ have i.i.d. quality that is $1/q$ with probability $q$ and $1$ otherwise, so all agents have the same prior mean of $2-q$. The two baseline policies illustrate the following trade-off:  

\begin{itemize}
    \item In the no-revelation policy, we assume the receiver allocates to an agent uniformly at random, since their posterior means are identical. This results in a total welfare of $2-q = O(1)$, while the max-min fair value achieved is $\frac{2-q}{n+1} = \Theta(1/n)$. 
\item  In the full-revelation policy, agent $0$ is chosen only when all other agents have value $1$, while happens with probability $O(e^{-\sqrt{n}})$. The social welfare is now $\Theta(\sqrt{n})$, which is a factor of $\Theta(\sqrt{n})$ larger than that of no-revelation. On the other hand, its max-min utility is now $O(e^{-\sqrt{n}})$, a super-polynomial factor worse than that of no-revelation.
\end{itemize}

Note that in our example, $V = 1/q$, so that the approximation ratio for majorization in \cref{thm:main-informal-exist} is $O(\log n/\epsilon)$. Clearly, the above two policies do not achieve this.

We now construct a policy that is simultaneously a constant-factor approximation to the social welfare of full-revelation and the max-min fair value of no-revelation. The sender designs a scheme for each  agent $i \in \{1,\dots,n\}$: if its true value is $1/q$, send a ``HIGH'' signal with a small probability $p=1/(nq)$; otherwise, send a ``LOW'' signal. This ensures the probability of any single agent sending a HIGH signal is exactly $q \cdot p = 1/n$. The receiver's posterior means are then:
\begin{itemize}
    \item $\mathbb{E}[v_i|\text{HIGH}_i] = 1/q = \sqrt{n}$, since the HIGH signal is only ever sent in the high-value state.
    \item For the LOW signal, we use Bayes' rule:
    \begin{align*}
        \mathbb{E}[v_i|\text{LOW}_i] &= \frac{\Pr(\text{LOW}|v_i=\frac{1}{q})\Pr(v_i=\frac{1}{q})\cdot\frac{1}{q} + \Pr(\text{LOW}|v_i=1)\Pr(v_i=1)\cdot 1}{\Pr(\text{LOW})} \\
        &= \frac{(1-p)q \cdot \frac{1}{q} + 1 \cdot (1-q)}{ (1-p)q + (1-q) } = \frac{1-p+1-q}{1-pq} = \frac{2-q-p}{1-pq}.
    \end{align*}
    Substituting $p=1/(nq)$, for large $n$ this posterior mean is $\frac{2-1/\sqrt{n}-1/n^{1.5}}{1-1/n}$, which is slightly smaller than $2-q$.
\end{itemize}
The receiver's strategy is as follows: if any HIGH signals are received, select one of these agents (posterior mean  $=\sqrt{n}$); if all signals are LOW (an event with constant probability for large $n$), select agent 0 (value $2-q$) over the others (posterior $\approx 2$). This policy achieves an expected social welfare of $\Theta(\sqrt{n})$, which is within a constant factor of the optimal welfare. At the same time, it guarantees a max-min utility of $\Theta(1/n)$, as agents $1,2,\ldots,n$ are selected with probability $\Theta(1/n)$ each. This single policy is therefore a constant-factor approximation to both the optimal social welfare (achieved by full revelation) and the optimal max-min utility (which is $\Theta(1/n)$).

\section{Impossibility of Majorization with Non-Polymatroidal Constraints}
\label{app:non-polymatroid-impossibility}


We construct a simple, deterministic allocation problem to show that for certain non-polymatroidal constraints, the approximation factor for majorization must grow at least linearly with the number of agents, showing that \cref{thm:main-informal-exist} cannot be generalized to arbitrary constraint sets, and requires special properties of polymatroids. Our counterexample holds when $V = 1$ and all value distributions $D_i$ are deterministic, so that no signaling is required.

Consider $n$ agents, each with a deterministic value of $v_i=1$. Utility is therefore equal to allocation. The set of feasible allocations $\mathcal{P}$ is the convex hull of $n$ vectors $\{u^{(1)}, \dots, u^{(n)}\} \subset \mathbb{R}^n$. For a large constant $M > n$, the vector $u^{(j)}$ is defined by its components $u^{(j)}_k = M^j$ if $k \ge j$ and $u^{(j)}_k=0$ if $k < j$. Any feasible allocation is a convex combination $x = \sum_{j=1}^n p_j u^{(j)}$ for some probability vector $p$. A crucial property of this construction is that any feasible allocation $x$ is sorted: $x_k = \sum_{j=1}^k p_j M^j = x_{k-1} + p_k M^k \ge x_{k-1}$. Thus, the $j$ smallest utilities are simply the first $j$ components of the allocation vector, and $Q_j(x) = \sum_{k=1}^j x_k$.

First, we find the optimal policy for each prefix sum objective. The objective $Q_j(x)$ is a linear function of the probabilities $p$, so its maximum must be achieved at a vertex of the probability simplex, i.e., by a pure policy $p_k=1$ for some $k$. If we choose the policy $p_k=1$, the allocation is $x=u^{(k)}$, and the prefix sum is $Q_j(u^{(k)}) = (j-k+1)M^k$ if $j \ge k$, and 0 otherwise. Since $M>n$, this value is maximized over $k \in \{1,\dots,j\}$ when $k=j$. Thus, the optimal policy for maximizing $Q_j$ is the pure strategy $p_j=1$, and the optimal value is $Q_j^* = M^j$.

Now, let us assume a single policy $x=\sum p_j u^{(j)}$ is $\beta$-majorized for some $\beta = o(n)$. This requires $Q_j(x) \ge Q_j^*/\beta = M^j/\beta$ for all $j \in \{1, \dots, n\}$. The exact value of the prefix sum is $Q_j(x) = \sum_{i=1}^j p_i M^i (j-i+1)$. We can bound this by isolating the dominant term: $Q_j(x) = p_j M^j + \sum_{i=1}^{j-1} p_i M^i (j-i+1)$. The summation is clearly bounded above by $n M^{j-1}$ (since $\sum p_i \le 1$ and $j-i+1 \le n$). The majorization condition thus implies $p_j M^j + n M^{j-1} \ge M^j/\beta$. Dividing by $M^j$, we get $p_j + n/M \ge 1/\beta$, which gives the necessary condition $p_j \ge 1/\beta - n/M$ for each $j$. Summing over all $j=1,\dots,n$:
\[ 1 = \sum_{j=1}^n p_j \ge \sum_{j=1}^n \left(\frac{1}{\beta} - \frac{n}{M}\right) = \frac{n}{\beta} - \frac{n^2}{M}. \]
Rearranging this gives a lower bound on the required approximation factor: $\beta \ge \frac{n}{1+n^2/M}$. By choosing $M$ to be a sufficiently large polynomial in $n$ (e.g., $M=n^3$), this implies $\beta \ge \Omega(n)$. This contradicts the assumption that $\beta$ is sub-linear. Therefore, no such policy can exist.

\section{Discussion on \cref{lem:mat-poly}}
\label{app:sec3}
We now present two pieces of evidence to show the non-triviality of \cref{lem:mat-poly}, in that it needs delicate arguments that are tailored to the setting we consider. First, we show that for non-polymatroidal allocation constraints, $1$-majorization in the allocation space will not imply majorization in the utility space, to any sub-linear approximation. This shows that our proof crucially requires the polymatroidal structure of the allocation space. We next show that part of our argument is not a generic result for polymatroids, and is specific to the utility polytope we define. In particular, we show that the result $\mathcal{B}(g) \subseteq \mathcal{U}$ is not merely a consequence of the submodularity of $g$ (where $g$ is the saturation function of $\mathcal{U}$), and the result is only true for the specific $\mathcal{U}$ (utility vectors of the receiver-optimal allocation) that we define. 

\paragraph{Impossibility with Majorized Allocation Set.} We first show a non-polymatroidal allocation set that has a $1$-majorized point, but even with deterministic values of the agents, the corresponding set of utility vectors does not have a point with sub-linear approximation to majorization. This rules out generalizing \cref{lem:mat-poly} to non-polymatroidal constraints.
 
In the example below, the value distributions $D_i$ are deterministic, so no signaling is required. Consider $2n$ agents, evenly partitioned into $n$ groups $g_1,\dots, g_n$, such that $g_i=\{i, 2n-i\}$. The feasible allocations is the convex hull of the indicator vectors $\{\mathbbm{1}[g_1],\dots, \mathbbm{1}[g_n]\}\subset \R^n$. The allocation that selects each $g_i$ with probability $1/n$ is a $1$-majorized policy. 

Let $N =\omega(n^{2n})$. For $1\leq i\leq n$, let agent $i$ take value $v_i=n^{2i}$ and agent $2n-i$ take value $v_{2n-i}=N-n^{2i}$. Since the total value of any group is $N$, the set of receiver-optimal allocations is the probability simplex $\{\vec{p} \mid \sum_{i=1}^n p_i = 1\}$, where $p_i$ is the probability of selecting group $g_i$. 

We first derive a lower bound for the sum of the smallest $i$ utilities. Consider the following class of allocations $\{A_i\}_{i=1}^n$. Define $A_i$ to be the allocation that selects $g_j$ with probability $1/n^2$ for $j\neq i$ and $g_i$ with probability $1-(n-1)/n^2$. Then, the $i^{th}$ smallest utility for $A_i$ is 
\[
\paren{1-\frac{n-1}{n^2}}\cdot n^{2i}\geq n^{2i}-n^{2i-1}.
\]
This then lower bounds the sum of the $i$ smallest utilities. 

Suppose next that there is an $\alpha$-majorized policy $A$ for the set of utility vectors, which selects $g_i$ with probability $p_i$. Then, the sum of $i$ smallest utilities is at most 
\[
\sum_{j=1}^{i-1}n^{2j}+p_i\cdot n^{2i}\leq \paren{p_i+\frac{2}{n^2}}\cdot n^{2i}. 
\]
For $A$ to be $\alpha$-majorized, 
\[
\alpha\cdot \paren{p_i+\frac{2}{n^2}}\cdot n^{2i}\geq n^{2i}-n^{2i-1}.
\]
Therefore, 
\[\alpha\cdot \paren{p_i+\frac{2}{n^2}}\geq 1-o(1), \qquad \mbox{and} \qquad p_i\geq \frac{1-o(1)}{\alpha}-\frac{2}{n^2}.\]
Since this holds for all $i$ and $\sum_{i=1}^n p_i=1$, we have
\[
\frac{n(1-o(1))}{\alpha}-\frac{2}{n}\leq 1.
\]
This implies 
\[
\alpha\geq \frac{n(1-o(1))}{1+\frac{2}{n}}=n(1-o(1)).
\]
Therefore, the best majorization factor for the set of utility vectors grows linearly with the number of agents, despite the existence of a $1$-majorized point for the set of allocation vectors. 

\paragraph{Submodular Saturation Functions do not Suffice.} We next show that the result $\mathcal{B}(g) \subseteq \mathcal{U}$ is not a general result for arbitrary convex polytopes $\mathcal{U}$ that are constant sum and  whose saturation function $g$ is submodular. Here, constant-sum means the sum of coordinates is a constant. Also recall that the saturation function $g(S)$ for a polytope is the maximum over the polytope of the sum of the coordinates in $S$. We show an example where a constant sum convex polytope $\mathcal{U}$ has submodular saturation function $g$, but is strictly contained in the base polytope $\mathcal{B}(g)$. This shows the proof of \cref{lem:mat-poly} is delicate in requiring specific properties of the $\mathcal{U}$ that we define. 

We start by defining the submodular function $g$. Let $h$ be defined as $h(0) = 0$, $h(1) = 1$, $h(2) = 1.9$, and $h(3) = 2.5$. This is clearly concave. Let $g(S) = h(|S|)$ for sets $S$ of size at most three. Consider the $3$-dimensional base polytope $\mathcal{B}(g)$. This is defined by the constraints:
$$ \left\{(x,y,z) \ge 0 \ | \ \max\{x,y,z\} \le 1; \ \max\{x+y, y+z, x+z\} \le 1.9; \ x+y+z = 2.5 \right\}. $$
Projecting onto the $(x,y)$ plane, we obtain the following hexagon:
$$ \left\{(x,y) \ge 0 \ | \ \max\{x,y\} \le 1; \ \min\{x, y\} \ge 0.6; \ 1.5 \le x+y \le 1.9 \right\}. $$
Note that there is a one-to-one mapping between the points in the planar hexagon and the points in $\mathcal{B}(g)$. Further, every edge in the planar hexagon corresponds to one saturation function of $\mathcal{B}(g)$. For instance, the edge $x+y = 1.5$ corresponds to $z = 1$ (corresponding to the set $S = \{3\}$, the dimension for $z$). Similarly, $x = 0.6$ corresponds to $y+z = 1.5$ (corresponding to the set $S = \{2,3\}$, the dimensions for $y,z$). The vertices of the hexagon are of the form $(a,b)$, where $a \neq b$, and $a,b \in \{1,0.9,0.6\}$.

Now consider some vertex of the hexagon, say $(1,0.6)$ and the corresponding vertex $v \in \mathcal{B}(g)$. Define $\mathcal{U} \subset \mathcal{B}(g)$ by intersecting the planar hexagon with a halfspace that removes $v$ from $\mathcal{B}(g)$, but preserves the other edges and vertices. For instance, add the halfspace $x - y \le 0.39$. Note that the resulting projection of the convex polytope $\mathcal{U}$ on the plane is defined by all of the original halfspaces of the planar hexagon, plus the new halfspace. This means that for any subset $S$ of dimensions, the function $g'(S)$ that maximizes the sum of the coordinates in $S$ over $\mathcal{U}$ coincides with the same function for $\mathcal{B}(g)$, which is $g(S)$. But $\mathcal{U} \subset \mathcal{B}(g)$ by construction. This also means $\mathcal{U}$ is constant sum, since $x+y+z = 2.5$ by construction. Therefore, we have an example $\mathcal{U}$ whose saturation function $g$ is submodular, but which lies strictly within $\mathcal{B}(g)$.

\end{document}